\documentclass[10pt,conference]{IEEEtran}
\usepackage{fixltx2e}
\usepackage{cite}
\usepackage{url}
\usepackage{mathrsfs}
\usepackage{color}
\usepackage{float}
\usepackage{balance}
\usepackage[caption=false]{subfig}

\ifCLASSINFOpdf
   \usepackage[pdftex]{graphicx}
   \graphicspath{{Figs/}}
   \DeclareGraphicsExtensions{.pdf,.jpeg,.png}
\else
\fi
\usepackage[cmex10]{amsmath}
\usepackage{amsmath}
\usepackage{amssymb}
\usepackage{amsthm}
\usepackage{amsfonts}
\usepackage{bm}
\usepackage{xfrac}
\usepackage{empheq}
\usepackage[normalem]{ulem} 
\usepackage{soul} 
\usepackage{mathtools}
\usepackage{nicefrac}
\usepackage{tikz}
\usepackage{tikzscale}
\usetikzlibrary{positioning}
\usetikzlibrary{shapes.geometric}
\usepackage{tikzpeople}

\usepackage{cleveref}

\newtheorem{theorem}{Theorem}

\newtheorem{example}{Example}

\newtheorem{corollary}{Corollary}
\newtheorem{remark}{Remark}
\theoremstyle{definition}

\renewcommand{\qed}{\hfill $\blacksquare$}
\newenvironment{sproof}{\noindent{ \emph{ Sketch of proof:}}}{\qed\bigskip}

\usepackage[a4paper,bindingoffset=0.2in,%
left=1in,right=1in,top=1in,bottom=1in,%
footskip=.25in]{geometry}
\graphicspath{{figs/}}
\begin{document}
	\newgeometry{left=0.7in,right=0.7in,top=0.7484252in,bottom=0.944882in}
    \title{On Information Theoretic Fairness: Compressed Representations With Perfect Demographic Parity} 
\vspace{-5mm}
\author{
		\IEEEauthorblockN{Amirreza Zamani, Borja Rodr{\'i}guez-G{\'a}lvez, and Mikael Skoglund \vspace*{0.5em}
			\IEEEauthorblockA{\\
                              Division of Information Science and Engineering, KTH Royal Institute of Technology \\
				Email: \protect amizam@kth.se, borjarg@kth.se, skoglund@kth.se }}
		}
	\maketitle

\begin{abstract}
    In this article, we study the fundamental limits in the design of fair and/or private representations achieving perfect demographic parity and/or perfect privacy through the lens of information theory. More precisely, given some useful data $X$ that we wish to employ to solve a task $T$, we consider the design of a representation $Y$ that has no information of some sensitive attribute or secret $S$, that is, such that $I(Y;S) = 0$.
    We consider two scenarios. First, we consider a design desiderata where we want to maximize the information $I(Y;T)$ that the representation contains about the task, while constraining the level of compression (or encoding rate), that is, ensuring that $I(Y;X) \leq r$. Second, inspired by the Conditional Fairness Bottleneck problem, we consider a design desiderata where we want to maximize the information $I(Y;T|S)$ that the representation contains about the task which is not shared by the sensitive attribute or secret, while constraining the amount of irrelevant information, that is, ensuring that $I(Y;X|T,S) \leq r$. 
    In both cases, we employ extended versions of the Functional Representation Lemma and the Strong Functional Representation Lemma and study the tightness of the obtained bounds. Every result here can also be interpreted as a coding with perfect privacy problem by considering the sensitive attribute as a secret.

\end{abstract}
\section{Introduction}
In this paper, we consider the common scenario illustrated in~\Cref{fig:ISITsys}, where we wish to employ some available data $X$ to make some decision or draw inferences about a certain task $T$. However, both the data and the task contain information about some sensitive attribute or secret $S$. For instance, the data $X$ could reflect a person's record in the census, the task $T$ if this person should be granted a loan or not, and the sensitive attribute $S$ the person's identity or their membership to a minority group. 

\looseness=-1 In this scenario, it is important to ensure that the decisions taken are not unfair and/or that the inferences made do not cause a privacy breach. One way to deal with this issue is to design a \emph{private} or \emph{fair representation} $Y$ of the data, that is, representations that are informative of the task $T$ but are uninformative of the sensitive attribute or secret $S$~\cite{gun, vari,zhao2022, zhao2019,zemel, gronowski2023classification,hardt, king3, borz, khodam, Khodam22,kostala,makhdoumi, yamamoto1988rate, sankar, Calmon2}. As discussed in~\cite{vari}, the privacy and fairness problems are very similar. In particular, we will only consider the design of representations $Y$ that contain no information about the sensitive attribute or secret $S$, that is, such that $I(Y;S) = 0$. This has different meanings in different communities.
In the \emph{information-theoretic privacy} literature, the independence of the disclosed representation $Y$ and the secret $S$ is known as \emph{perfect privacy} or \emph{perfect secrecy}~\cite{borz, kostala2}.
In the \emph{algorithmic fairness} literature, the independence of the algorithm's output and the sensitive attribute is known as having a zero \emph{demographic (or statistical) parity gap} \cite{vari,zhao2022,zhao2019,zemel}. Here, we will refer to this condition as \emph{perfect demographic parity}.

\looseness=-1 Under this strong requirement on the fairness and/or privacy of the representations, we consider two problems that further trade off utility and compression or encoding rate:
\begin{enumerate}
    \item In order to guarantee that the representation $Y$ has as much predictive power as possible, we want to maximize the information it contains about the task $T$. Moreover, we impose a minimum level of compression $r$ to the data representation, that is, $I(Y;X) \leq r$. This optimization problem is described in~\eqref{eq:problem_1}.
    \item Inspired by the \emph{Conditional Fairness Bottleneck (CFB)}~\cite{vari}, we attempt at maximizing the information the representation $Y$ keeps about the task $T$ that is not shared by the private data $S$. Moreover, we impose a maximum amount of bits $r$ allocated to the irrelevant information about the data $X$ which is not shared by the private data $S$ and the task $T$, that is, $I(Y;X|S,T) \leq r$. This optimization problem is illustrated in~\Cref{fig:cfb-perfect} and described in~\eqref{eq:problem_2}.
\end{enumerate}

\begin{figure}[t]
    \centering
    \begin{tikzpicture}
        \draw[draw=black, fill=gray!5, rounded corners] (-2.8, -1.8) rectangle ++(6.2,4);
        \node at (2.9, -1.5) {agent};
     
        \node[draw, circle] (secret) at (0,1.2) {$S$};
        \node[draw, circle] (data) at (0,0) {$X$};
        \node[draw, circle] (task) at (0, -1.2) {$T$};
        \node[draw, rectangle, rounded corners] (joint) at (-2,0) {$P_{S,X,T}$};
        \node[draw, rectangle, rounded corners] (mechanism) at (2.5,0) {$P_{Y|S,X,T}$};
        \node[draw, circle] (representation) at (4.4,0) {$Y$};
        
        \node at (0, 1.8) {\begin{tabular}{c} \small{sensitive attribute} \\ \small{or secret} \end{tabular}};
        \node at (0.0, 0.55) {\small{data}};
        \node at (0, -0.7) {\small{task}};
        \node at (2.5, 0.8) {\begin{tabular}{c} \small{fairness} \\ {or privacy} \\ \small{mechanism} \end{tabular}};
        \node at (4.4, 0.65) {\begin{tabular}{c} \small{disclosed} \\ \small{representation} \end{tabular}};

        \draw[-stealth] (joint.east) -- (secret.west);
        \draw[-stealth] (joint.east) -- (data.west);
        \draw[-stealth] (joint.east) -- (task.west);

        \draw[-stealth] (data.east) -- (mechanism.west);
        \draw[dashed] (secret.east) -- (mechanism.west);
        \draw[dashed] (task.east) -- (mechanism.west);

        \draw[-stealth] (mechanism.east) -- (representation.west);
    \end{tikzpicture}
    \caption{\looseness=-1 Data representation with perfect demographic parity or privacy. We want to design a representation $Y$ of the data $X$ that is useful for the task $T$, is compressed, and independent of the sensitive attribute or secret $S$.}
    \label{fig:ISITsys}
\end{figure}
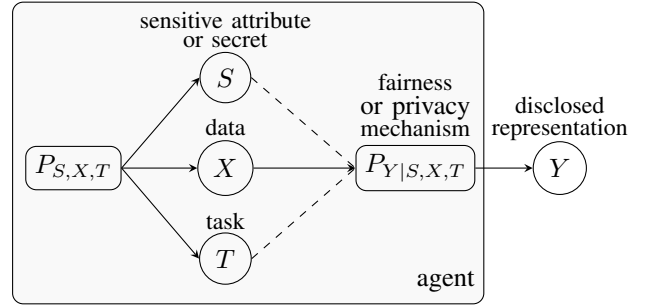

\begin{figure}[t]
	\centering
	\includegraphics[width=0.4\textwidth]{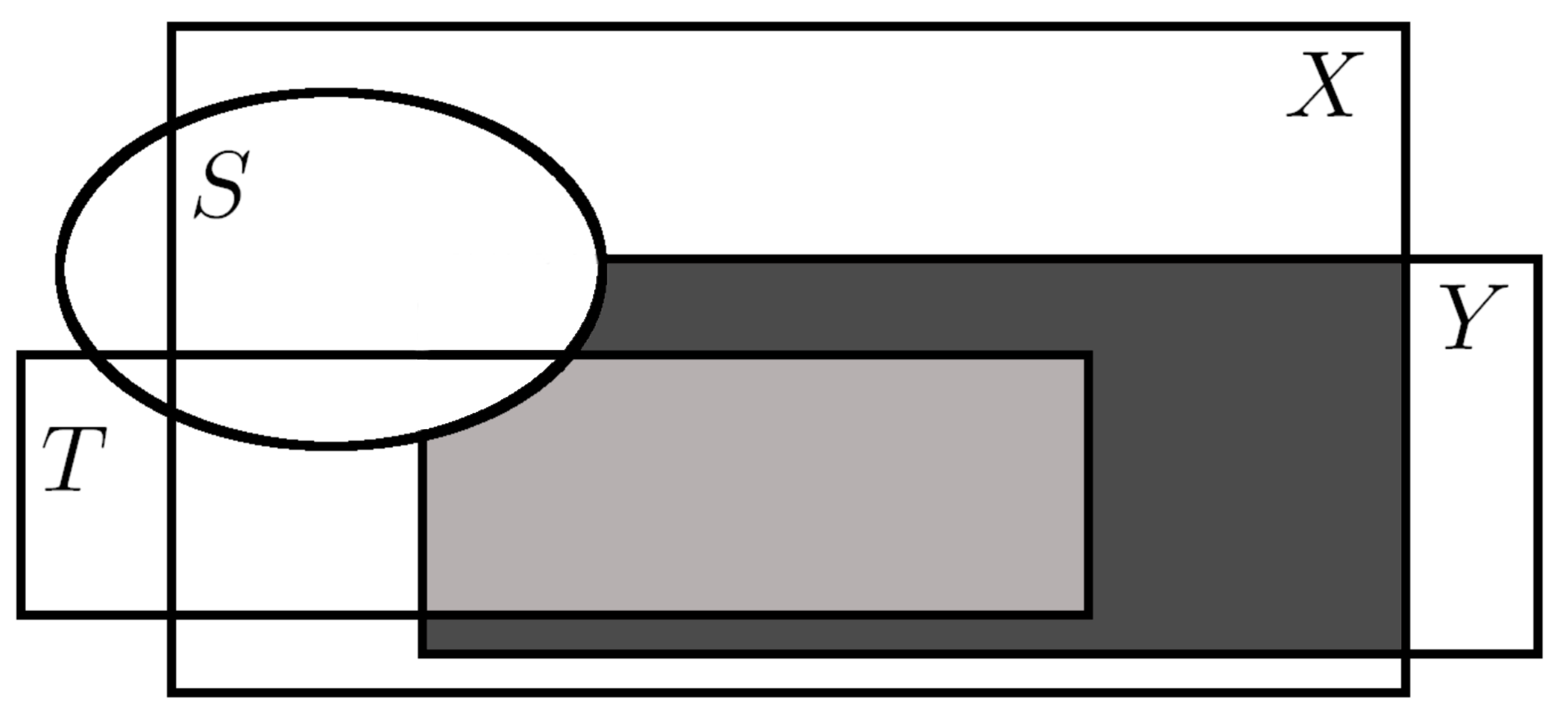}
	\caption{Information diagram~\cite{yeung1991new} of the CFB with perfect demographic parity. In light gray, we show the relevant information about the task $T$, not shared by the sensitive attribute $S$, that we want the representation $Y$ to maximize. In dark gray, we show the irrelevant information about the data $X$ that we want to constrain. Contrary to the standard CFB~\cite{vari}, we enforce that the representation contains no information about the sensitive attribute (perfect demographic parity or perfect privacy).}
	\label{fig:cfb-perfect}
\end{figure}

\subsection{Prior work on fairness mechanism design}

The notion of \emph{fair representations} was introduced by Zemel et al.~\cite{zemel}. This advanced the field of algorithmic fairness due to the expressiveness of deep learning and is mainly dominated by adversarial~\cite{zemel, edwards2016censoring, zhao2019} and variational~\cite{vari, creager2019flexibly, louizos2015variational, gupta2021controllable} approaches. The trade offs between utility and fairness were theoretically studied in~\cite{zhao2022}.

In \cite{vari}, the authors study the problem from an information-theoretic perspective. They introduce the CFB, which trades off the utility, fairness and \emph{compression} of the representations in terms of the mutual information. This extra desideratum has since then been also considered in subsequent work~\cite{gupta2021controllable, de2022funck, gronowski2023classification}.

\subsection{Prior work on privacy mechanism design}

In \cite{borz}, the authors study the trade-off between privacy and utility in terms of the mutual information. They show that the optimal privacy mechanism under perfect privacy can be obtained as the solution of a linear problem. In \cite{khodam, Khodam22}, this is subsequently generalized by relaxing the perfect privacy assumption (allowing some small bounded privacy leakage).
The concept of the \emph{Privacy Funnel} is introduced in~\cite{makhdoumi}, where the trade-off between privacy and utility considers the log-loss as the privacy and distortion metrics, ultimately leading again to the mutual information. An extension of this problem, the \emph{Conditional Privacy Funnel}, was introduced in~\cite{vari}. 

In \cite{Calmon2}, the authors studied the fundamental limits of the privacy and utility trade-off under an estimation-theoretic framework. In \cite{yamamoto, sankar}, the trade-off between privacy and utility are studied considering the equivocation as the privacy measure and the expected distortion as utility. In \cite{kostala}, the authors introduce the concept of \emph{secrecy by design} and apply it to privacy mechanism design and lossless compression problems. They find bounds on the trade-off between perfect privacy and utility using the \emph{Functional Representation Lemma} (FRL). In \cite{king3}, these results are extended to arbitrary privacy leakage. In \cite{kostala}, the authors also studied the average length of the encoded message, and in~\cite{kostala2}, they studied the relationships between the shared key, secrecy, and compression considering perfect secrecy, secrecy by design, maximal leakage, mutual information, and differential privacy for lossless compression.
 
\subsection{Contributions}

In this paper, we propose a simple and constructive theoretical approach to design fair representations with perfect demographic parity and/or private representations with perfect privacy. We consider two problems depending on the definition of the utility and compression constraint, which we outlined above.

\looseness=-1 The design of the mechanism yielding the representations is based on the extensions of the FRL and the Strong FRL (SFRL) derived in~\cite{king3}. In~\cite{king3}, the authors show that both lemmas are constructive and simple, helping them to obtain mechanism designs that can be optimal in special cases. We follow a similar path to \cite{king3} and \cite{kostala} to obtain simple mechanisms for fair and private representations. Finally, we compare the obtained designs in different scenarios and show that they are optimal in special cases.

\section{System model and Problem Formulation} 
\label{sec:system}
Let the sensitive attribute or secret $S \in \mathcal{S}$, data $X \in \mathcal{X}$, and target (task) $T \in \mathcal{T}$, be discrete random variables with joint distribution $P_{S,X,T} \in \mathbb{R}^{|\mathcal{S}|\times |\mathcal{X}| \times |\mathcal{T}|}$ and marginals $P_S \in \mathbb{R}^{|\mathcal{S}|}$, $P_X \in \mathbb{R}^{|\mathcal{X}|}$, and $P_T \in \mathbb{R}^{|\mathcal{T}|}$. Moreover, let also the data representation (or disclosed data) $Y \in \mathcal{Y}$ be a discrete random variable. As mentioned in the introduction, we seek to design a mechanism $P_{Y|S,X,T} \in \mathbb{R}^{|\mathcal{Y}| \times |\mathcal{S}|\times |\mathcal{X}| \times |\mathcal{T}|}$ that returns fair representations $Y$ with perfect demographic parity (or private representations with perfect privacy), that is, independent of the sensitive attribute or secret $S$.

\subsubsection{Problem 1} Here, the goal is to maximize the information the representation $Y$ has about the task $T$ while maintaining a minimum level of compression $r$. More precisely, the design problem is described by 
\begin{align}
    g^{r}(P_{S,X,T})&=\sup_{\begin{array}{c} 
	\substack{P_{Y|S,X,T}:I(Y;S)=0,\\ \ \ \ \ \ \ \ \  I(X;Y)\leq r }
	\end{array}}I(Y;T).
    \label{eq:problem_1}
\end{align} 
\looseness=-1 We consider the case where $r < H(X)$. Otherwise this leads to the privacy-utility trade-off problem studied in~\cite{kostala, king3}.

\begin{remark}
    \label{remark1}
	\normalfont
    \looseness=-1 The trade-off \eqref{eq:problem_1} is equivalent to the privacy-utility trade-off with a rate constraint studied in \cite{gun}, where they additionally incorporate the constraint that the Markov chain $(S,T)-X-Y$ holds. They show that when the rate constraint $r$ is sufficiently large ($r \geq H(X)$), this constraint is inactive, can be ignored, and the problem leads to a linear program. In contrast to \cite{gun}, in this paper, we obtain new mechanism designs with active rate constraints and without the Markov chain assumption. A simple scenario where this Markov chain naturally holds and can be ignored in the optimization is when $S$ and $T$ are deterministic functions of $X$.    
\end{remark}

\subsubsection{Problem 2} Here, we draw the utility and compression desiderata from the CFB~\cite{vari} with the additional requirement of perfect demographic parity of perfect privacy. The goal is to maximize the information the representation $Y$ has about the task $T$ which is not shared by the sensitive attribute or secret $S$ while keeping a bounded amount of irrelevant information. In this work, the term \emph{irrelevant} refers to the information about the data $X$ that does not reveal information about the task $T$ and the private data $S$. More precisely, the design problem is described by 
\begin{align}
g^{r}_c(P_{S,X,T})&=\sup_{\begin{array}{c} 
	\substack{P_{Y|S,X,T}:I(Y;S)=0,\\ \ \ \ \ \ \ \ \ \ I(X;Y|S,T)\leq r }
	\end{array}}I(Y;T|S).
    \label{eq:problem_2}
\end{align} 

\begin{remark}
	\normalfont
In \eqref{eq:problem_2}, in contrast to \cite{vari}, we do not require the Markov chains $T-X-Y$ and $S-X-Y$ and assume that $X$, $S$ and $T$ are arbitrarily correlated. As in~\Cref{remark1}, when $S$ and $T$ are deterministic functions of $X$ both Markov chains hold naturally and can be ignored in the optimization.
\end{remark}

\section{Main Results}
\label{sec:resul}
In this section, we investigate the fundamental limits of Problems 1 and 2. That is, we derive upper and lower bounds on the utilities of these problems under perfect demographic parity or perfect privacy and a bounded encoding rate, namely $g^{r}(P_{S,X,T})$ and $g^{r}_c(P_{S,X,T})$ defined in \eqref{eq:problem_1} and \eqref{eq:problem_2}. For this purpose, similarly to~\cite{king3}, we use the extended versions of the FRL and the SFRL proposed in~\cite[Lemmata 4 and 5]{king3}. These lemmata have a constructive proof, and thus describe a mechanism to construct the representation $Y$.

In order to bound $g^{r}(P_{S,X,T})$ and $g^{r}_c(P_{S,X,T})$, similarly to~\cite{kostala, king3}, we state an important expression for the utility $I(Y;T)$. Namely, we have that 
\begin{align}
I(Y;T)&=I(X,S,T;Y)-I(X,S;Y|T),\nonumber\\&=I(X,S;Y)+I(T;Y|X,S)-I(X,S;Y|T),\nonumber\\&=I(X,S;Y)+H(T|X,S)-H(T|Y,X,S)\nonumber\\ &\ \quad -I(X,S;Y|T).
\label{key}
\end{align}

\subsection{Bounds for Problem 1}
\label{subsec:main_p1}

Equipped with~\eqref{key}, we provide lower and upper bounds on $g^{r}(P_{S,X,T})$ in the next theorem.

\begin{theorem}
    \label{th.1}
	For every compression level $0\leq r\leq H(X|S)$ and random variables $(S, X, T) \sim P_{S,X,T}$, we have that
	\begin{align*}
	\max\{L_1^{r},L_2,L_3^{r}\}\leq g^{r}(P_{S,X,T})\leq H(T|S),
	\end{align*}
	where
	\begin{align*}
	L_1^{r} &= H(T|X,S)+r-H(X,S|T),\\
	L_2 &= H(T|X,S)-\left( \log(I(X,S;T)+1)+4 \right),\\
	L_3^{r} &= H(T|X,S)+r-\alpha H(X,S|T)-4\\&\!\!\!\!\!\!\!\! \ -\!\log((1\!-\!\alpha)I(X,S;T)\!+\!\alpha\min\{H(T),H(X,S)\}\!+\!1),
\end{align*}
and $\alpha=\nicefrac{r}{H(X|S)}$. Moreover, for $H(X|S)\leq r< H(X)$ we have that
 \begin{align}
 \label{th22}
 \max\{{L'}_1^{r},L_2\}\leq g^{r}(P_{S,X,T})\leq H(T|S),
 \end{align}
 where $
 {L'}_1^{r} = H(T|S)-H(S|T)=H(T)-H(S).$
 
\end{theorem}
\begin{sproof}
	The complete proof is provided in Appendix A. By removing the rate constraint we obtain the upper bound. To derive $L_2$, we use the SFRL \cite[Theorem~1]{kosnane} with $X\leftarrow (X,S)$, $Y\leftarrow T$, and the output $Z\leftarrow Y$. Using \eqref{key}, we have that
	\begin{align*}
	I(T;Y)&=H(T|X,S)-I(X,S;Y|T)\\&\geq H(T|X,S)\!-\!\left( \log(I(X,S;T)+1)+4 \right)\!=\!L_2.
	\end{align*} 
	The key idea to achieve $L_1^{r}$ and $L_3^{r}$ is to randomize over the output $U$ of the FRL \cite[Lemma1]{kostala} to produce $U'$. To obtain $L_1^{r}$, let $U$ be the output of the FRL. First, we show that $I(U;X)\leq H(X|S)$. We then construct $U'$ using the randomization technique from \cite[Lemma 4]{king3}. Let 
    \begin{equation*}
        U'=
        \begin{cases}
	       U,\ \text{w.p}.\ \alpha\\
	       c,\ \ \text{w.p.}\ 1-\alpha
	   \end{cases},
    \end{equation*}
    where $\alpha=\nicefrac{r}{H(X|S)}$. Furthermore, let $Y'$ be produced by the FRL with $S\leftarrow (S,X,U')$ and $X\leftarrow T$. Finally, let $Y=(U',Y')$. We first show that, for such $Y$, the constraints $I(S;Y)=0$ and $I(X;Y)\leq r$ are satisfied. Then, by using \eqref{key}, we obtain $L_1^{r}$. More precisely, the key equations for this result are
	\begin{align*}
	&I(Y',U';X,S)\!=\!I(U';X,S)\!=\!I(U';S)\!+\!I(U';X|S)\!=\!r,\\
	&H(T|X,S,U',Y')=0.
	\end{align*} 
	To derive $L_3^{r}$, we use same $U'$, but to produce $Y'$ we use the conditional SFRL with $S\leftarrow (S,X)|U'$ and $X\leftarrow T$. Next, we find an upper bound on $I(X,S;Y',U'|T)$ in \eqref{key}. Using this upper bound and the fact that $I(X,S;Y',U')=r$ we obtain the lower bound $L_3^{r}$. Finally, to derive ${L'}_1^{r}$, we use the following inequality
	\begin{align*}
	g  ^{r}(P_{S,X,T})\geq g^{r=H(X|S)}(P_{S,X,T}).
	\end{align*}
\end{sproof}

\begin{corollary}
	In the regime $H(X|S)\leq r < H(X)$, if the sensitive attribute or secret is a deterministic function of the task $S=f(T)$, the lower bound ${L'}_1^{r}$ is tight.
\end{corollary}

\begin{remark}
	\normalfont
    In \eqref{eq:problem_1}, we enforced perfect demographic parity as our definition of fairness. However, we may have employed similar techniques to those in~\Cref{th.1} to enforce perfect equalized odds, which is achieved imposing the conditional independence of the representation and the sensitive attribute, given the task, that is, that $I(Y;S|T) = 0$. Due to the limited space we leave this as a future work. 
\end{remark}
\begin{remark}
	\normalfont
	The lower bounds in~\Cref{th.1} can also be used to address the rate limited privacy-utility trade-off from~\cite{gun} with $r<H(X)$.
\end{remark}
\begin{remark}
	\normalfont
	In order to derive lower bounds in the regime $H(X|S)\leq r < H(X)$, we used the inequality $g^{r}(P_{S,X,T})\geq g^{r=H(X|S)}(P_{S,X,T})$. 
	This bound can also be achieved by using $Y=(U,Y')$, where $U$ is the output of the FRL, since in this regime $U$ satisfies that $I(U;X)\leq H(X|S)\leq r$ and $I(U;X,S)=H(X|S)$.
	Moreover, by letting $r=H(X|S)$, we have $L_1^{r}\geq L_3^{r}$. Hence, $L_3^{r}$ is not considered in \eqref{th22}. 
\end{remark}

The key idea to achieve $L_1^{r}$ and $L_3^{r}$ is to randomize over the output $U$ of the FRL to produce $U'$. To see the reason for producing $U'$, we focus on the term $I(Y;X,S)$ in \eqref{key}. 
The challenge is that, using the constraints $I(Y;S)=0$ and $I(Y;X)\leq r$, we have that
\begin{align*}
	I(Y;X,S)=I(Y;S)+I(Y,X|S)\leq H(X|S)
\end{align*}
and somehow we want to make this term as large as possible while satisfying the rate constraint. By producing $U$, we have 
\begin{align*}
    I(U;X,S)=I(U;S)+I(U,X|S)=H(X|S)
\end{align*}
\looseness=-1 since $I(U;S)=0$ and $H(X|U,S)=0$. However, considering the regime $r\leq H(X|S)$, the output $U$ does not necessarily satisfy that $I(U;X)\leq r$. Hence, we randomize over $U$ so that $I(U';X)$ becomes smaller or equal to $r$ at the expense of decreasing the utility from $H(X|S)$ to $r$, that is, $I(U;X,S)=H(X|S)$ and $I(U';X,S)=r$.  
Thus, the main reason to produce $U'$ is to achieve a non-zero utility for the term $I(U';X,S)$ in \eqref{key}, while satisfying $I(U';X)\leq r$. Furthermore, we produce $Y'$ to ensure that the term $H(T|X,S,Y)$ in \eqref{key} is zero. Note that $Y'$ does not change the term $I(Y;X,S)$ since it is independent of $(U',X,S)$.
	
\subsubsection{Comparison and Discussion}

 Here, we compare the obtained bounds in~\Cref{th.1}.
To compare $L_1^{r}$ and $L_2$, consider the case where $r$ is small and $H(X,T|S)\gg4$, then $L_2\geq L_1^{r}$. On the other hand, if $H(X,T|S)\leq 4$, then $L_2\leq L_1^{r}$. To compare $L_2$ and $L_3^{r}$, let $H(X|S)\leq H(X,S|T)$, then $L_2\geq L_3^{r}$ since $ r-\alpha H(X,S|T)\leq 0$ and
\begin{equation*}
    I(X,S;T)\leq (1\!-\!\alpha)I(X,S;T)\!+\!\alpha\min\{H(T),H(X,S)\}.
\end{equation*}
A simple example where $H(X|S)\leq H(X,S|T)$ is when $T$ is a deterministic function of $S$. On the other hand, let $H(X|S) \geq H(X,S|T)$. Ignoring the $\log(\cdot)$ terms (since they have smaller values compared to the other terms) we can have $L_2\leq L_3^{r}$. A simple example where this happens ($H(X|S) \geq H(X,S|T)$) is when $S$ is a deterministic function of $T$. Next, we compare $L_1^{r}$ and $L_3^{r}$. To do so, consider a strong compression requirement, that is $0 < r \ll H(X|S)$, and $H(X,S|T)\gg 4$, then we can have $L_1^{r}\leq L_3^{r}$. In other words, comparing $L_3^{r}$ and $L_1^{r}$, we have decreased the term $H(X,S|T)$ in $L_1^{r}$ by multiplying it by $0\leq\alpha\leq1$, but we get a penalty that corresponds to $\log((1\!-\!\alpha)I(X,S;T)\!+\!\alpha\min\{H(T),H(X,S)\}\!+\!1)\!+\!4$. Intuitively, when $r$ is small, then $L_3^{r}$ is larger than $L_1^{r}$, but when $r$ grows $L_1^{r}$ becomes larger. Considering the extreme case when $r=H(X|S)$, we see that $L_1^{r}\geq L_3^{r}$. When $r\geq H(X|S)$, it seems that ${L'}_1^{r}$ is a good candidate since its distance with the upper bound is $H(S|T)$. Furthermore, to achieve it, considering \eqref{key}, the term $I(Y;X,S)$ achieves $H(X|S)$, which is its maximum with the constraints $I(Y;S)=0$ and $I(X;Y)\leq r$. 
\begin{remark}
	\normalfont
	The advantage of $L_1^{r}$ and $L_3^{r}$ compared to $L_2$ is that they achieve a non-zero utility for the term $I(Y;X,S)$, that is, the added term corresponding to the encoding rate constraint $r$ in both bounds. For instance, consider a case where $T$ is a deterministic function of $X$. In this case, $L_2\leq 0$. However, both $L_1^{r}$ and $L_3^{r}$ can achieve non-zero utility.   
\end{remark}

\begin{example}
	Let $r=H(X|S)$ and $H(S|T)<4$. In this case, $L_1^{r}$ is dominant and we have $L_1^{r}\geq \max\{L_2^{r},L_3^{r}\}$. This follows since
	\begin{align*}
		&L_1^{r=H(X|S)}-L_2^{r=H(X|S)}=H(T|S)-H(T|X,S)\\&\qquad \quad \quad+ \log(I(X,S;T)+1)+4 -H(S|T)\geq 0.
	\end{align*} 
	Furthermore,
	\begin{align*}
		&L_1^{r=H(X|S)}-L_3^{r=H(X|S)}=\\&\qquad \qquad \quad \log(\min\{H(T),H(X,S)\}+1)+4\geq 0.
	\end{align*} 
\end{example}

\subsection{Bounds on Problem 2}

Similarly to what we did in~\Cref{subsec:main_p1}, here we derive upper and lower bounds on $g^{r}_c(P_{S,X,T})$ defined in \eqref{eq:problem_2}.

\begin{theorem}
    \label{th.2}
    For every compression level $H(X|T,S) \leq r $ and random variables $(S,X,T) \sim P_{S,X,T}$, we have that
    \begin{align*}
		g^{r}_c(P_{S,X,T})=H(T|S).
	\end{align*}
	Furthermore, if $\log(I(X;T|S)+1)+4\leq r< H(X|T,S)$, then
	\begin{align}
    \label{chap}
		L_1\leq g^{r}_c(P_{S,X,T})\leq H(T|S)
	\end{align}
	where $L_1 = H(T|S,X)-\log(I(X;T|S)+1)-4$.
	\end{theorem}   
\begin{proof}
	The upper bound can be easily derived since $I(Y;T|S)\leq H(T|S)$. When $r\geq H(X|T,S)$, the bounded irrelevant information constraint is satisfied since $I(X;Y|S,T)\leq H(X|S,T)$. Furthermore, let $Y$ be found by the FRL \cite[Lemma 1]{kostala} using $S\leftarrow S$ and $X\leftarrow T$. In this case, we have that
	\begin{equation}
	I(Y;S)=0 \textnormal{ and } H(T|Y,S)=0.
    \label{mama}
	\end{equation}  
	Furthermore, equation~\eqref{mama} means that $$I(T;Y|S)=H(T|S).$$ Hence, considering $r\geq H(X|T,S)$ we have that $g^{r}_c(P_{S,X,T})=H(T|S)$. Next, let $\log(I(X;T|S)+1)+4\leq r\leq H(X|T,S)$. We build $Y$ using the conditional SFRL from~\cite{kosnane}. Thus, $Y$ satisfies that
	\begin{align}
	I(Y;S,X)&=0,\label{3ek}\\
	H(T|S,X,Y)&=0, \textnormal{ and} \label{3ek1}\\
	I(Y;X|T,S)&\leq \log(I(X;T|S)+1)+4.\label{3ek2}
	\end{align}
	Next, we expand $I(Y;X,T|S)$ in two ways. We have that
	\begin{align}
	I(Y;X,T&|S)=I(Y;T|S)+I(Y;X|S,T) \nonumber \\  &\leq I(Y;T|S)+\log(I(X;T|S)+1)+4,\label{esha2}
	\end{align}
	where we used \eqref{3ek2}.
	On the other hand, we have that
	\begin{align}
	I(Y;X,T|S)&=I(Y;X|S)+I(Y;T|X,S)\nonumber\\&=H(T|X,S),\label{esha}
	\end{align}
	where we used \eqref{3ek} and \eqref{3ek1}. Combining \eqref{esha2} and \eqref{esha}, we conclude the proof noting that
    \begin{align*}
	   I(T;Y|S)\geq H(T|X,S)-\log(I(X;T|S)+1)-4.
	\end{align*}
\end{proof}

\begin{corollary}
	When the data $X$ is a deterministic function of the sensitive attribute $S$ or task $T$, we have $g^{r}_c(P_{S,X,T})=H(T|S)$, since in this case $r\geq H(X|T,S)=0$.
\end{corollary}

\begin{remark}
    The results in~\Cref{th.2} are obtained for a limited level of compression requirements, that is, for rates $r$ larger than a certain threshold. The problem where the compression requirement is arbitrarily strong is left to the extended version of the paper. 
\end{remark}
 
\section{Conclusion}
In this paper, we studied the fundamental limits to generate representations $Y$ of some data $X$ that are informative of a certain task $T$ and are fair (have perfect demographic parity) and/or private (have perfect privacy). That is, the limits of generating representations that are independent of a sensitive attribute or secret $S$. Moreover, the representations are required to have a certain level of compression (or a bounded encoding rate $r$). 

\looseness=-1 More precisely, we studied two different problems. In Problem 1, the utility is measured by the mutual information $I(Y;T)$ the representation keeps about the task, and the compression requirement is the standard encoding rate $I(X;Y)$. In problem 2, we draw the utility and the compression requirements from the Conditional Fairness Bottleneck. Namely, the utility is measured by the mutual information $I(Y;T|S)$ the representation keeps about the task that is not shared by the sensitive attribute or secret, while the compression requirement is a limit on the information $I(X;Y|S,T)$ the representation keeps about the data that is not shared by the task and the sensitive attribute or secret.

To design these representations, we used a randomization technique on top of the Functional Representation and Strong Functional Representation Lemmata. These lemmata are constructive and give us both theoretical bounds on the problems and a mechanism to obtain the representations. Therefore, the combination of this randomization technique with the lemmata leads to simple representation designs. In both problems, when the rate $r$ is larger than a threshold, the obtained designs can be optimal. Specifically, in Problem 1, when the sensitive attribute or the secret is a deterministic function of the task, the lower bound is tight.

\section*{Appendix A}
\textit{Proof of Theorem 1:} To derive the upper bound, by removing the rate constraint we have that
\begin{align*}
g^{r}(P_{S,X,T})\leq \sup_{\begin{array}{c} 
	\substack{P_{Y|S,X,T}:I(Y;S)=0}
	\end{array}}I(Y;T) \leq H(T|S),
\end{align*}
where the last inequality holds by using \cite[Theorem 7]{kostala}. Next, to derive $L_2$, we use the SFRL \cite[Theorem~1]{kosnane}. Let $Y$ be the outcome of the SFRL with $X\leftarrow (X,S)$, $Y\leftarrow T$, and the output $Z\leftarrow Y$. In this way, we have that
\begin{align}
I(X,S;Y)&=0,\label{s}\\
H(T|Y,X,S)&=0, \textnormal{ and} \label{s1}\\
I(X,S;Y|T)&\leq \log(I(X,S;T)+1)+4. \nonumber
\end{align} 
Clearly, the outcome $Y$ satisfies the constraints $I(Y;S)=0$  and $I(Y;X)=0\leq r$. Then, using~\eqref{key} we have that
\begin{align*}
I(T;Y)&\stackrel{(a)}{=}H(T|X,S)-I(X,S;Y|T)\\&\geq H(T|X,S)\!-\!\left( \log(I(X,S;T)+1)+4 \right)\\&=L_2,
\end{align*}
where in (a) we used \eqref{s} and \eqref{s1}. To obtain $L_1^{r}$ ,let $U$ be the output of the FRL \cite[Lemma~1]{kostala}. Then, we have that
\begin{align}
I(S;U)&=0,\label{ss}\\
H(X|U,S)&=0, \textnormal{ and} \label{ss1} \\
I(U;X)&\leq H(X|S),\label{tt}
\end{align}
where~\eqref{tt} holds since
\begin{align*}
I(U;X)&\stackrel{(a)}{=}I(U;S)+H(X|S)-H(X|U,S)-I(S;U|X)\\&\stackrel{(b)}{\leq} H(X|S).
\end{align*}
Step (a) holds for any correlated $X$, $S$, and $U$ and in step (b) we used \eqref{ss} and \eqref{ss1}. Now, we construct $U'$ using the randomization technique that is used in \cite[Lemma 4]{king3}. Let $$U'=\begin{cases}
U,\ \text{w.p}.\ \alpha\\
c,\ \ \text{w.p.}\ 1-\alpha
\end{cases},$$ 
\looseness=-1 where $c$ is a constant which does not belong to $\mathcal{X}\cup\mathcal{S}\cup\mathcal{U}$ and $\alpha=\nicefrac{r}{H(X|S)}$. Furthermore, let $Y'$ be produced by the FRL with $S\leftarrow (S,X,U')$ and $X\leftarrow T$. Hence, we have that
\begin{align}
I(Y';S,X,U')&=0 \textnormal{ and} \label{sss}\\
H(T|Y',S,X,U')&=0.\label{ssss}
\end{align}
Finally, let $Y=(U',Y')$. We first check the constraints $I(S;Y)=0$ and $I(X;Y)\leq r$. We have that
\begin{align*}
I(S;Y)&=I(S;Y',U')\\&\stackrel{(a)}{=}I(S;U')\\&=\alpha I(S;U)=0,
\end{align*}
where (a) follows by \eqref{sss}. Moreover, it holds that
\begin{align*}
I(X;Y)&=I(X;Y',U')\\&\stackrel{(a)}{=}I(X;U')\\&=\alpha I(X;U)\\ &\stackrel{(b)}{\leq} \alpha H(X|S)=r,
\end{align*}
where (a) follows from \eqref{sss} and (b) follows from \eqref{tt}. Next, by using \eqref{key} we have that
\begin{align*}
I(T;Y)&=I(T;Y',U') \\
&\stackrel{(a)}{=}H(T|X,S)\!+\!I(Y',U';X,S)\!-\!I(X,S;Y'\!,U'|T) \\
&= H(T|X,S)+I(U';X,S)-I(X,S;Y',U'|T) \\
& \geq H(T|X,S)+I(U';X,S)-H(X,S|T)
\\& \stackrel{(b)}{=}H(T|X,S)+I(U';X|S)-H(X,S|T)\\
&\stackrel{(c)}{=}H(T|X,S)+r-H(X,S|T)
\end{align*}
where in (a) we used \eqref{ssss} and in (b) we used that $I(U';S)=0$. Furthermore, step (c) holds since
\begin{align*}
I(U';X|S)=\alpha I(U;X|S)=\alpha H(X|S)=r,
\end{align*}
where we used \eqref{ss1}. 

To derive $L_3^{r}$, we use same $U'$, but to produce $Y'$ we use the conditional SFRL with $S\leftarrow (S,X)|U'$ and $X\leftarrow T$. 
In this case, $U'$ satisfies \eqref{sss}, \eqref{ssss} and
\begin{align*}
&I(Y';S,X|T,U') \\
&\leq \log(I(X,S;T|U')+1)+4\\
&= \log((1-\alpha)I(X,S;T)+\alpha I(X,S;T|U)+1)+4\\
&\!\leq \log((1\!-\!\alpha)I(X,S;\!T)\!+\!\alpha\min\{H(X,S),H(T)\}\!+\!1\!)\!+\!4.
\end{align*}	
Next, we find an upper bound on $I(X,S;Y',U'|T)$ in \eqref{key}. We have that
\begin{align}
I(&X,S;Y',U'|T)\nonumber \\
&=I(X,S;U'|T)+I(X,S;Y'|T,U')\nonumber\\
&=\alpha I(X,S;U|T)+I(X,S;Y'|T,U') \nonumber \\
&\leq \alpha H(X,S|T)+\log((1\!-\!\alpha)I(X,S;\!T) \nonumber\\ & \quad +\!\alpha\min\{H(X,S),H(T)\}\!+\!1\!)\!+\!4.\label{t}
\end{align}
Thus, using \eqref{key}, \eqref{t}, \eqref{ssss}, and the fact that $I(Y;S,X)=r$, we have that $I(T;Y) \geq L_3^r$, where
\begin{align*}
    L_3^r = &H(T|X,S)+r-\alpha H(X,S|T) \\
    &-\log((1\!-\!\alpha)I(X,S;\!T)\! \\
    &+\!\alpha\min\{H(X,S),H(T)\}\!+\!1\!)\!+\!4.
\end{align*}
Finally, to derive ${L'}_1^{r}$, note that for any $H(X|S)\leq r< H(X)$, we have that
\begin{align*}
g^{r}(P_{S,X,T})&\geq g^{r=H(X|S)}(P_{S,X,T})\\&\geq L_1^{r=H(X|S)}\\&=H(T|X,S)+H(X|S)-H(X,S|T)\\&= H(T,X|S)-H(S,X|T)\\&= H(T|S)-H(S|T).
\end{align*}
\hfill\IEEEQED

\bibliographystyle{IEEEtran}
{\balance \bibliography{IEEEabrv,ITW2024}}
\end{document}